\pdfoutput=1
\documentclass[pdftex,pra,aps,twoside,nofootinbib,showkeys,showpacs,twocolumn,superscriptaddress]{revtex4}
\usepackage{setspace}
\usepackage{float}
\usepackage{mciteplus}
\usepackage{indentfirst}
\usepackage{url}
\usepackage[stable]{footmisc}
\usepackage{fancyhdr}
\usepackage{appendix}
\usepackage{epstopdf}
\usepackage[all]{xy}
\usepackage{ifpdf}
\usepackage{color}
\ifpdf
 \usepackage[pdftex]{graphicx}
 \else
\usepackage{graphicx}
 \fi
\usepackage{graphicx}
\usepackage[T1]{fontenc}
\usepackage{textcomp}
\usepackage{newcent}
\usepackage[sc,noBBpl]{mathpazo}
\usepackage{amsmath,amsfonts,amssymb,amsthm, graphics}
\usepackage[mathscr]{eucal}
\usepackage{appendix}
\usepackage{bbm,bm}

\theoremstyle{plain}
\newtheorem{theorem}{Theorem}

\newtheorem{fact}{Fact}

\newtheorem{rslt}[theorem]{Result}
\newtheoremstyle{note}{\topsep}{\topsep}{\slshape}{}{\scshape}{}{ }{}
\theoremstyle{note}

\newcommand{\mbV}{\mathbb{V}}

\newcommand\tr{\operatorname{Tr}}

\newcommand{\<}{\langle}
\renewcommand{\>}{\rangle}
\newcommand\be{\begin{equation}}
\newcommand\ee{\end{equation}}
\newcommand\bea{\begin{array}}
\newcommand\eea{\end{array}}
\newcommand\ben{\begin{eqnarray}}
\newcommand\een{\end{eqnarray}}
\newcommand\ot{\otimes}

\newcommand\bei{\begin{itemize}}
\newcommand\eei{\end{itemize}}
\newcommand\bee{\begin{enumerate}}
\newcommand\eee{\end{enumerate}}
\newcommand{\ket}[1]{| #1 \rangle}
\newcommand{\bra}[1]{\langle #1 |}

\newcommand{\E}{\operatorname{e}}

\newcommand{\mathsym}[1]{{}}
\newcommand{\unicode}[1]{{}}

\def\<{\langle}
\def\>{\rangle}
\def\ot{\otimes}

\newcommand{\la}{\langle}
\newcommand{\ra}{\rangle}

\newcommand{\pd}{\partial}

\newcommand{\ii}{\mathbf{i}}
\newcommand{\hc}{\dagger}
\newcommand{\SP}{\text{ }}

\newsavebox{\smlmat}% Box to store smallmatrix content
\savebox{\smlmat}{$\left[\begin{smallmatrix}p&\alpha \\\alpha^{*}&\widetilde{p}\end{smallmatrix}\right]$}
\begin{document}
\title{Generic appearance of objective results in quantum measurements}

\author{J. K. Korbicz}
\email{jkorbicz@mif.pg.gda.pl}
\affiliation{Faculty of Applied Physics and Mathematics, Gda\'nsk University of Technology, 80-233 Gda\'nsk}
\affiliation{National Quantum Information Centre in Gda\'nsk, 81-824 Sopot, Poland}
\author{E. A. Aguilar}
\affiliation{Faculty of Mathematics, Physics and Informatics University of Gda\'nsk, 80-952 Gda\'nsk, National Quantum Information Centre in Gda\'nsk, 81-824 Sopot, Poland}
\author{P. \'Cwikli\'nski}
\affiliation{Faculty of Mathematics, Physics and Informatics University of Gda\'nsk, 80-952 Gda\'nsk, National Quantum Information Centre in Gda\'nsk, 81-824 Sopot, Poland}
\author{P. Horodecki}
\affiliation{Faculty of Applied Physics and Mathematics, Gda\'nsk University of Technology, 80-233 Gda\'nsk}
\affiliation{National Quantum Information Centre in Gda\'nsk, 81-824 Sopot, Poland}
%\author{J. K. Korbicz$^{1,2}$, E. A. Aguilar$^{3}$, P. \'Cwikli\'nski$^{4}$, and P. Horodecki$^{1,2}$}
%\affiliation{
%$^1$ Faculty of Applied Physics and Mathematics, Gda\'nsk University of Technology, 80-233 Gda\'nsk\\
%$^2$ National Quantum Information Centre in Gda\'nsk, 81-824 Sopot, Poland  \\
%$^3$ Institute of Mathematics, University of Gda\'nsk, 80-952 Gda\'nsk, National Quantum Information Centre %in Gda\'nsk, 81-824 Sopot, Poland
%$^4$ Institute of Theoretical Physics and Astrophysics, University of Gda\'nsk, 80-952 Gda\'nsk, National %Quantum Information Centre in Gda\'nsk, 81-824 Sopot, Poland
%}
%\thanks{jkorbicz@mif.pg.gda.pl}

\date{\today}

\begin{abstract}
Measurement is of central interest in quantum mechanics as it provides the link between the quantum world
and the world of everyday experience. One of the features of the latter is its robust, objective character, contrasting the delicate nature of quantum systems.
Here we analyze in a completely model-independent way the celebrated von Neumann measurement process, using recent techniques of information flow, studied in open quantum systems.
We show the generic appearance of objective results in quantum measurements, provided
we macroscopically coarse-grain the measuring apparatus and wait long enough. To study genericity, we employ the widely-used Gaussian Unitary Ensemble
of random matrices and the Hoeffding inequality.
We derive  generic  objectivization timescales, given solely by the interaction strength and the systems' dimensions. 
Our results are manifestly universal and are a generic property of von Neumann measurements.
\end{abstract}

\maketitle

Understanding quantum measurements has been one of the central problems of quantum theory since its beginning  
\cite{Bohr1949-BOHDWE,heisenberg1952philosophic}. It not only provides the crucial link between the theory and experiment, the micro- and macro-worlds, 
but is at the heart of the modern quantum technologies (see e.g. \cite{measurement2009}).
The fundamental measurement theory dates back to von Neumann \cite{Neumann_book} and 
since then has been further developed in various directions, e.g. the decoherence theory \cite{Schlosshauer_book, joos2003decoherence}.
To be readable, measurement results must inevitably be encoded into macroscopic degrees of freedom
and one of the crucial features expected from a good measurement process is an objective character of the results: They can be read out  
by arbitrary many observers and without causing any disturbance by the mere read-out. 
This has been realized as early as in 1929 by Mott \cite{Mott}.
Achieved in well engineered measurements by a proper coupling to macroscopic degrees of freedom, 
it is not at all obvious if such a situation is a generic feature of a quantum measurement process with a macroscopic recording.

In a broader context of open quantum systems \cite{Schlosshauer_book, joos2003decoherence}, 
this may be seen as a question about how information flows from the system to its environment. %(the macroscopic measuring apparatus). 
Pioneering research along this direction has been undertaken under the 
quantum Darwinism idea \cite{Zurek2009_darwinism}, 
arguing that in some situations (see e.g. \cite{PhysRevLett.105.020404, PhysRevLett.101.240405}) perfect information about the system can be redundantly stored in  the environment
and becomes effectively classical \cite{Ollivier2004} and objective. The generic character of some of the quantum Darwinism features 
was shown in \cite{Brandao2015_darwinism} and the universality of decoherence was shown on short time-scales 
in \cite{Braun2001_decoherence, Strunz2003_decoherence, Strunz2003_1_decoherence, Yukalov2012_decoherence}. 
A further step was recently made in \cite{Korbicz2014_Objectivity,Horodecki2015_objectivity} by formulating information
flow and objectivity in the fundamental language of quantum states with the
introduction of the, so called, Spectrum Broadcast Structures (SBS's).
The latter has been proven to be a useful tool 
allowing to obtain novel results  in some of the emblematic models of decoherence
\cite{Korbicz2014_Objectivity, Tuziemski2015_objectivity, Tuziemski2016_sbs, MKH16}. 
Finally, questions of genericity have traditionally been the domain of statistical mechanics and  thermodynamics (see e.g.  \cite{Popescu-thermo, Masanes2011-thermo, Brandao2011-thermo}).
Phrased in this language, we may ask to what form a generic state equilibrates during a von Neumann measurement. % and how it encodes the results. 

In this communication we study information flow during a  von Neumann measurement process with a macroscopic (in a sense of a number of degrees of freedom)
measuring apparatus. Applying  random matrix theory techniques \cite{Mehta_book, Haake_book},
we show that generically the  post-measurement state approaches, after a coarse-graining, a form carrying almost perfect, multiple records of the measurement result,
thus making the latter objective. To study genericity, we use a
properly structured  Gaussian Unitary Ensemble (GUE) \cite{Mehta_book,Haake_book}.
Since the seminal works of Wigner and Dyson on statistics of various experimentally observed spectra, it has been the basic choice for random Hamiltonians
due to its universality and agreement with the experiment \cite{Mehta_book,Haake_book}.
The apparatus is assumed to be noisy, with the initial state distributed according to some physically motivated 
measures of mixed states \cite{Sommers2004_purity}. For large-dimensional measured systems,
we provide estimates on the time-scale of the objectivization process. %As a by-product, we also show the typical non-Markovian character of the evolution.
Since the only assumptions we make concern the genericity measures, our results are manifestly  universal and apply to the whole class of von Neumann measurements,
thus showing a generic and robust character of the emergence of objectivity. It is a bit of a surprise that this property of von Neumann measurements was so far tacitly assumed
(see e.g. \cite{SEWELL2005271}) but never, to our best knowledge, derived.

{\bf Measurements with compound apparatuses.--} (cf. \cite{Korbicz2014_Objectivity})
We consider  a $d_S$-dimensional quantum system $S$ simultaneously measured by a collection of $N$ measuring apparatuses/environments 
%(we will use those terms interchangeably) 
$E_1$,\ldots, $E_N$, each of dimension $d$, representing a macroscopic measuring device.
The apparatuses are assumed to be individually coupled to the system through a general von Neumann-type interaction, so strong
that the self Hamiltonians of the system and the apparatuses can be neglected (the quantum measurement limit) \cite{Neumann_book}:
\be
\label{ham}
\hat H_{total}\approx\hat H_{int}= \hat{A} \ot \sum_{k=1}^{N} \hat{B}_k,
\ee
where  $\hat{A}$ is the measured observable (assumed non-degenerate) 
and the $\hat{B}_k$ are some general measuring observables. 
%In the quantum measurement limit studied here
%this interaction Hamiltonian dominates the dynamics and the self Hamiltonians of the system and the apparatuses can be neglected
%\cite{Neumann_book}. 
This leads to the evolution (setting $\hbar=1$) 
$\hat{U} \equiv \E^{-i t\hat{H}_{int}}= \sum_a |a\> \<a| \ot \bigotimes_{k=1}^{N} \E^{-i a\hat{B}_k t}$, 
where $\hat A=\sum_{a=1}^{d_S} a \ket a \bra a$. 
Our main object of study is a partially
reduced state $\rho_{S:E_{obs}}$, with  a fraction $E_{uno}$ of size $N_{uno}$ of unobserved subsystems traced out. 
This represents an inevitable loss of information during a measurement.% (e.g. we do not see all the photons scattered off a given object).
Assuming 
$\rho_{SE}(0)=\rho_{0S}\ot\bigotimes_{k=1}^N \rho_{0k}$ we obtain:
\begin{eqnarray}\label{ptr}
&&\rho_{S:E_{obs}}(t)=\sum_a p_a |a\> \<a| \ot \bigotimes_{k=1}^{N_{obs}} \rho_{ak}(t)+ \sum_a\sum_{a'\ne a} c_{aa'}\\
&&\times \left\{ \prod_{k=1}^{N_{uno}} \tr [\E^{- i (a-a')\hat{B}_k t} \rho_{0k}]\right\}|a\> \<a'| \bigotimes_{k=1}^{N_{obs}} \E^{-i a\hat{B}_k t}\rho_{0k} \E^{i a'\hat{B}_k t}, \nonumber
\end{eqnarray}
%\tr_{E_{uno}}\left[ \hat{U} \rho_{0S}\ot\bigotimes^N_{k=1} \rho_{0k} \hat{U}^\hc \right]
where $p_a\equiv \langle a|\rho_{0S}|a\rangle$, $c_{aa'}\equiv \<a|\rho_{0S}|a'\>$, $\rho_{ak}(t)\equiv \E^{-i a\hat{B}_k t}\rho_{0k} \E^{i a\hat{B}_k t}$ , $N_{uno}+N_{obs}=N$.
We define the decoherence factor for the unobserved fraction $E_{uno}$:  
\be
\label{eq:deco2}
\Gamma^{uno}_{aa'}(t) \equiv \prod_{k=1}^{N_{uno}} \left|\tr [\E^{- i (a-a')\hat{B}_k t} \rho_{0k}]\right|^2.
\ee
If for all $a\ne a'$: i) $\Gamma^{uno}_{aa'}(t)= 0$, i.e. decoherence takes place, and 
ii)  $\rho_{ak}(t)\perp\rho_{a'k}(t)$, i.e. $\rho_{ak}(t)$ are perfectly distinguishable, then we say that $\rho_{S:E_{obs}}(t)$ is of a  Spectrum Broadcast Structure (SBS) 
\cite{Korbicz2014_Objectivity, Horodecki2015_objectivity,Tuziemski2015_objectivity} 
with respect to (w.r.t.) the basis $\ket a$ (this context-dependence is of a fundamental importance, see e.g. \cite{Auffeves2016_onto}), defined as \cite{PhysRevA.86.042319}:
\be\label{SBS}
\rho_{SBS}=\sum_a p_a \ket a\bra a\otimes \rho_a\otimes\dots\otimes\rho_a,\  \rho_a\perp\rho_{a'\ne a}.
\ee
The basis $\ket a$ becomes then the, so-called, pointer basis in which the system has decohered and the result of the measurement, $a$, appearing with the probability $p_a$, 
is stored in the measuring setup in many,  perfect  copies. Crucially, their readouts, through projections on the supports of $\rho_{ak}(t)$, will not  disturb (on average) the joint state $\rho_{S:E_{obs}}(t)$.
This leads to a form of objectivity of the measurement result: It can be read out by multiple observers without disturbing neither the (decohered) system nor themselves
\cite{Zurek2009_darwinism, Korbicz2014_Objectivity, Horodecki2015_objectivity}. 
In quantum-information terms, this objectivization process is a weaker form of quantum state broadcasting \cite{no-broadcast, PianiPRL2008}. 
We can thus reformulate the original question as: Are SBS's generic for the interactions \eqref{ham}?
To address it, we introduce an ensemble of random Hamiltonians of the form \eqref{ham} and random initial conditions $\rho_{0k}$.
We then estimate the average trace distance between the actual state \eqref{ptr} and an ideal SBS in the following steps:
i) calculate the  averages over $\hat B_k$ of the decoherence factor \eqref{eq:deco2} and the, so called, super-fidelity bound
for the states $\rho_{ak}(t)$; ii) average them over $\rho_{0k}$; iii) coarse-grain the apparatus; iv) further average over $\hat A$;
v) use the central result of \cite{MKH16} to bound the average distance and show that it vanishes in the macroscopic limit.
We then use the concentration inequality of Hoeffding \cite{Hoe}, following from the classical Chernoff bound, to show genericity.

The coarse-graining is one of the crucial steps. As we will show, on the microscopic level of 
the individual apparatuses, the residual noise is too strong to allow a SBS formation even asymptotically. 
This can be overcome if we group the $N_{obs}$ observed apparatuses into fractions scaling with $N$ (called macrofractions) 
and pass to the thermodynamic limit $N\to\infty$ \cite{Korbicz2014_Objectivity}.
The number $\mathcal M$ of such groups (assumed for simplicity equal) is irrelevant, provided their sizes $N_{mac}\equiv N_{obs}/\mathcal M$ satisfy 
$N_{mac}\sim N$. These macrofractions may be understood as reflecting some detection threshold, 
e.g.  a minimum bunch of photons the eye can detect.

{\bf Randomizing measurement Hamiltonians.-- }
We introduce an ensemble of random measurement Hamiltonians \eqref{ham} using the widely-used Gaussian Unitary Ensemble \cite{Haake_book,Mehta_book}
in the following way (cf. \cite{GPKS, CGS}): i) $\hat B_k$ are independently, identically distributed (i.i.d.) according to a GUE with a scale factor $\eta_E$; 
ii)  $\hat A$ is distributed according to its own GUE with a scale factor $\eta_S$. We recall that the GUE measure is defined as:
\be
\label{eq:dis}
{\rm d}\mu_{gue} (\hat{H})=\frac{1}{Z}\E^{-\frac{\eta}{2}\sum_i\lambda_i^2}\prod_{i<j}(\lambda_i-\lambda_j)^2{\rm d}\pmb{\lambda}{\rm d} {\cal \hat{U}},
\ee
with $Z$ the normalization, $\lambda_i$ the eigenvalues, $\eta$ a scale factor, and
${\rm d}{\cal \hat{U}}$ the Haar measure on the unitary group.  

The simultaneous vanishing of the decoherence factor \eqref{eq:deco2} and 
of the generalized overlaps \cite{Jozsa-fidelity,no-broadcast} 
$F_{aa'}\equiv$ $F(\rho_a, \rho_{a'}) \equiv (\tr \sqrt{\sqrt{\rho_a}\rho_{a'}\sqrt{\rho_a}})^2$ for all $a\ne a'$ has so far been used to witness a SBS formation  \cite{Korbicz2014_Objectivity, MKH16}.
The latter function is however complicated and here we will use the so-called super-fidelity bound \cite{Miszczak2009_fidelity}
$F(\rho, \sigma) \leq G(\rho,\sigma)\equiv\tr\left(\rho \sigma\right) + \sqrt{(1-\tr \rho^2)(1-\tr \sigma^2)}$
(although we note that it is not tight if both states are mixed, as e.g. for 
$\rho\perp\sigma$, $G(\rho,\sigma)\ne 0$), which here  reads:
\be
\label{fidbound}
G\left(\rho_a(t), \rho_{a'}(t)\right)=\tr\left(\rho_a(t) \rho_{a'}(t)\right)+  S_{lin}(\rho_0)\equiv G_{aa'}(t),
\ee
where $S_{lin}(\rho_0)\equiv 1-\tr \rho_0^2$ is the linear entropy of the initial state of an individual apparatus.

We now average (\ref{eq:deco2},\ref{fidbound}) over the interaction and the initial conditions.
%Since both functions are non-negative, the vanishing of the averages will imply that fluctuations above zero are typically small.
%and hence an SBS typically appears, leading to the objectivity of the measurement result.
We first average over $\{\hat B_k\}$, fixing the levels $a,a'$ of $\hat A$. We have:
\be\label{iid}
\langle \Gamma^{uno}_{aa'}(t)\rangle_{\{\hat B_k\}}=\prod_{k=1}^{N_{uno}}\langle |\tr [\E^{- i (a-a')\hat{B}_k t} \rho_{0k}]
|^2\rangle_{\hat B_k},
\ee
since $\hat B_k$ are i.i.d. Modulo $\rho_{0k}$, all the factors
are identical and  we calculate the average over a single $\hat B_k$, dropping the index $k$ for simplicity. 
Performing the Haar integration first (\cite{SM}, Section IA) and then the eigenvalue one (\cite{SM}, Section IIB), 
we obtain \cite{nota}:

\begin{rslt}\label{eav}
The GUE averages of the single environment decoherence  and super-fidelity factors read:
\begin{eqnarray}
&&\< \Gamma_{aa'}(t) \>=  \frac{1+ \tr \rho^2_0 }{d+1}+\< f_t(\pmb a,\pmb{\lambda})\> \frac{2(d - \tr \rho^2_0)}{d(d^2-1)}, \label{eq:GUEavg}\\
&&\< G_{aa'}(t) \>=S_{lin}( \rho_0)+ \frac{1+ \tr \rho^2_0 }{d+1}\nonumber\\
&&+\< f_t(\pmb a,\pmb{\lambda})\>  \frac{2(d\tr \rho^2_0-1)}{d(d^2-1)}, \label{eq:GUEavgB}
\end{eqnarray}
with $f_t(\pmb a,\pmb{\lambda}) \equiv \sum_m\sum_{n>m} \cos\left[(a-a')(\lambda_n-\lambda_m)t\right]$ and:
\begin{eqnarray}\label{fav}
&&\< f_t(\pmb a,\pmb{\lambda})\> = p(d,\tilde{\Delta}_t) \E^{-\tilde{\Delta}_t^2},\\
&&p(d,\tilde{\Delta}_t)\equiv \label{pd}\\
&&\sum_n\sum_{m>n} \left[
L_n^{(0)}(\tilde{\Delta}_t^2)L_m^{(0)}(\tilde{\Delta}_t^2) -
\frac{n!}{m!}\tilde{\Delta}_t^{2(m-n)}
[L_n^{(m-n)}(\tilde{\Delta}_t^2)]^2 \right]\nonumber
% \nonumber\\
%&&\sum_{n<m} \Bigg[ \sum_{k=0}^{n} \sum_{l=0}^m \binom{n}{k} \binom{m}{l} \frac{(-1)^{n+m-k-l}
%\tilde{\Delta}_t^{2(n+m-k-l)}}{(n-k)!(m-l)!}\nonumber\\
%&&- \sum_{k=0}^{n} \sum_{l=0}^n \binom{n}{k} \binom{m}{l} \frac{(-1)^{k+l}\tilde{\Delta}_t^{2(n+m-k-l)}}{(m-k)!(n-l)!}\Bigg],\label{pd}
\end{eqnarray}
where $\tilde\Delta_t \equiv (a-a')t/\sqrt{\eta_E}$ and $L_n^{(m)}$ are the associated Laguerre polynomials.
%given by (\ref{fav},\ref{pd}) and $\<\tr \rho^2_0\>$ by \eqref{eq:HS}.
\end{rslt}

The above results are exact. Although 
the average $\< f_t(\pmb a,\pmb{\lambda})\>$ with the GUE eigenvalue distribution $P_{gue}(\pmb{\lambda})$
involves only  the two-point correlation function
\cite{Mehta_book}: $R_2(\lambda_1,\lambda_2)\equiv d!/(d-2)!\int\cdots\int {\rm d}\lambda_3\cdots{\rm d} \lambda_d
P_{gue}(\lambda_1,\ldots,\lambda_d)$ (due to the symmetry), and
the large-$d$ asymptotics of  $R_2(\lambda_1,\lambda_2)$ are well known  \cite{Mehta_book}, 
they are of no use here. One can show that \cite{Mehta_book}: $R_2(\lambda_1,\lambda_2)=
K_d(\lambda_1,\lambda_1)K_d(\lambda_2,\lambda_2)-[K_d(\lambda_1,\lambda_2)]^2$, 
$K_d(\lambda_1,\lambda_2)\equiv\sum_{j=0}^{d-1}\phi_j(\lambda_1)\phi_j(\lambda_2)$, with $\phi_j(\lambda)$ the oscillator wave-functions, and
while the first term  approaches the Wigner semicircle distribution,
integrable with $f_t(\pmb a,\pmb{\lambda})$, the second term approaches a function of
$|\lambda_1-\lambda_2|$ only \cite{Mehta_book} and makes the integral divergent. That is the integration
with $f_t(\pmb a,\pmb{\lambda})$ and the large-$d$ limit are not interchangeable here. 

Both (\ref{eq:GUEavg}, \ref{eq:GUEavgB}) depend on $\rho_0$ only through its purity
$\tr\rho_0^2$ and  we can use the known results of generic state purity to effectively get rid of the
initial state dependence. Although there is no canonical choice of a measure over mixed states,
there are several popular ones e.g. the Hilbert-Schmidt and the Bures measures \cite{Sommers2004_purity} giving:
\begin{eqnarray}
\label{eq:HS}
\< \tr \rho^2_0 \>_{HS} = \frac{2d}{d^2+1},\  \< \tr \rho^2_0 \>_{Bu} = \frac{5d^2+1}{2d(d^2+2)}.
\end{eqnarray}
Especially the Bures measure is physically important as it: i) is directly connected to quantum metrology \cite{metrology};
ii) reproduces the correct measure for pure states.  
In what follows we will assume that $\rho_{0k}$ are i.i.d. with one of the above measures and are averaged over.
%turns the products over the environments in  (\ref{gprod},\ref{fprod}) into respective powers
%of single-copy GUE averages $\<\cdot \>_{gue}$ , depending on $\< \tr \rho^2_E \>$.

{\bf Residual noise and coarse-graining.--} As $p(d,\tilde{\Delta}_t)$ is an even polynomial of degree $2(2d-3)$, 
\eqref{fav} implies that the time dependent part in (\ref{eq:GUEavg}, \ref{eq:GUEavgB})   
decays for any fixed  $d$ and  a gap $|a-a'|\ne 0$  with a characteristic time $\tau_{aa'}\equiv |a-a'|^{-1}\sqrt{\eta_E/(d+1)}$ (\cite{SM}, Section II).
The remaining constant terms: 
A common one of the order $O(1/d)$ (cf. (\ref{eq:HS})), called "white noise", 
and additionally $\<S_{lin}( \rho_0)\>$ in \eqref{eq:GUEavgB}. 
The  latter, arising from the non-tight bound (\ref{fidbound}), is intuitively understood---the noisier the apparatus is initially, the lesser information, measured by
the state distinguishability, it can accumulate.
These factors, reflecting residual background fluctuations in the ensemble, pertain to a single apparatus and prevent a SBS formation.
However, coming back to \eqref{eq:deco2}, using \eqref{iid} and (\ref{eq:HS}), we actually obtain an exponential decay
with $N_{uno}$ of the collective decoherence factor:
\be
0\leq\< \Gamma_{aa'}^{uno}(t) \>=\< \Gamma_{aa'}(t) \>^{N_{uno}}\xrightarrow[t\gg\tau_{aa'}]{} O\left(d^{-N_{uno}}\right), \label{gas}
\ee
showing that for a large local dimension $d$ and/or large unobserved fraction $N_{uno}$, measurement dynamics \eqref{ham} generically leads
to decoherence (cf. \cite{Braun2001_decoherence}).  The same step can be preformed on the observed fraction too \cite{Korbicz2014_Objectivity}:
We group the $N_{obs}$ observed apparatuses into $M$ groups of $N_{mac}$ each, described by states 
$\rho_a^{mac}(t)\equiv\bigotimes_{k\in mac}\rho_{ak}(t)$.
Due to the factorization of fidelity w.r.t. the tensor product
and the i.i.d. property, the resulting super-fidelity bound \eqref{fidbound} for the group also decays (cf. (\ref{eq:HS})): 
\be
0\leq \<F^{mac}_{aa'}(t)\>\leq\< G_{aa'}(t) \>^{N_{mac}}\xrightarrow[t\gg\tau_{aa'}]{} O\left(\E^{-\frac{N_{mac}}{d}}\right)\label{ggas}.
\ee
If both $N_{uno}, N_{mac}$ scale with $N$,
(\ref{gas},\ref{ggas}) can be made small in the macroscopic/thermodynamic limit $N\to\infty$.
Crucially, increasing $d$ alone is not enough--it damps the white noise, but
$\<S_{lin}( \rho_0)\>\simeq 1-O(1/d)$ by (\ref{eq:HS}). %Thus, we are led to consider macrofractions.

{\bf Generic post measurement state and objectivity.--}
Results (\ref{eq:GUEavg}, \ref{eq:GUEavgB}) still depend on the $\hat A$'s level differences $|a-a'|$.
%This is the only dependence on $\hat A$, the basis in \eqref{SBS} is kept fixed.
To study a completely general behavior, a further averaging of 
$\< \Gamma_{aa'}^{uno}(t) \>$, $\< G_{aa'}^{mac}(t) \>$ over
the levels $a,a'$  should be performed with the corresponding two-point correlation function $R_2(a,a')$
(the average is independent of the labels $a,a'$ due to the symmetry). 
The resulting integrals are intractable, but from
(\ref{eq:GUEavg}-\ref{pd}) they will eventually reach 
the noise-floor (see Fig. \ref{fig:exact}).
Lower bounds on the relevant timescales can be obtained from a short-time analysis (\cite{SM}, Section II), giving for the decoherence and 
the superfidelity respectively: 
\begin{eqnarray}
\begin{split}
&&\tau_{dec}\equiv \left[8g^2N_{uno}d_S\left(d-\<\tr\rho_0^2\>\right)\right]^{-\frac{1}{2}}\stackrel{d\gg 1}{\sim}
\frac{g^{-1}}{\sqrt{N_{uno}d_Sd}},\\
&&\tau_{fid}\equiv \left[8g^2N_{mac}d_S\left(d\<\tr\rho_0^2\>-1\right)\right]^{-\frac{1}{2}}\stackrel{d\gg 1}{\sim}
\frac{g^{-1}}{\sqrt{N_{mac}d_S}}.\label{tfid}
\end{split}
\end{eqnarray}
Here $g^{-1}\equiv\sqrt{\eta_S\eta_E}$ is the effective interaction time-scale and $d_S$ is the system dimension. 
We see a characteristic separation of time-scales: From \eqref{eq:HS}, 
$\tau_{fid}\sim\sqrt d\tau_{dec}$ for the same macrofraction sizes. Thus, on average, it takes longer to
accumulate information in the apparatus than to decohere the system \cite{Korbicz2014_Objectivity, Tuziemski2016_sbs}.
Combining (\ref{tfid}) with Result \ref{eav} and (\ref{gas},\ref{ggas}) we arrive at (cf. \cite{distinguishability}):
\begin{rslt}\label{result2} The interaction and initial state averages satisfy:
\begin{eqnarray}
&&\<\<\Gamma^{uno}(t)\>\>\xrightarrow[t\gg\tau_{dec}]{}O\left(\E^{-N_{uno}\log d}\right), \label{r1}\\
&&\<\<\ G^{mac}(t)\>\>\xrightarrow[t\gg\tau_{fid}]{}O\left(\E^{-\frac{N_{mac}}{d}}\right).\label{r2}
\end{eqnarray}
\end{rslt}

%----------FIGURE -------
%-------------------------
%-------------------------
\begin{figure}[t]
\centering
\begin{minipage}{.5\columnwidth}
  \centering
\includegraphics[width=0.9\textwidth]{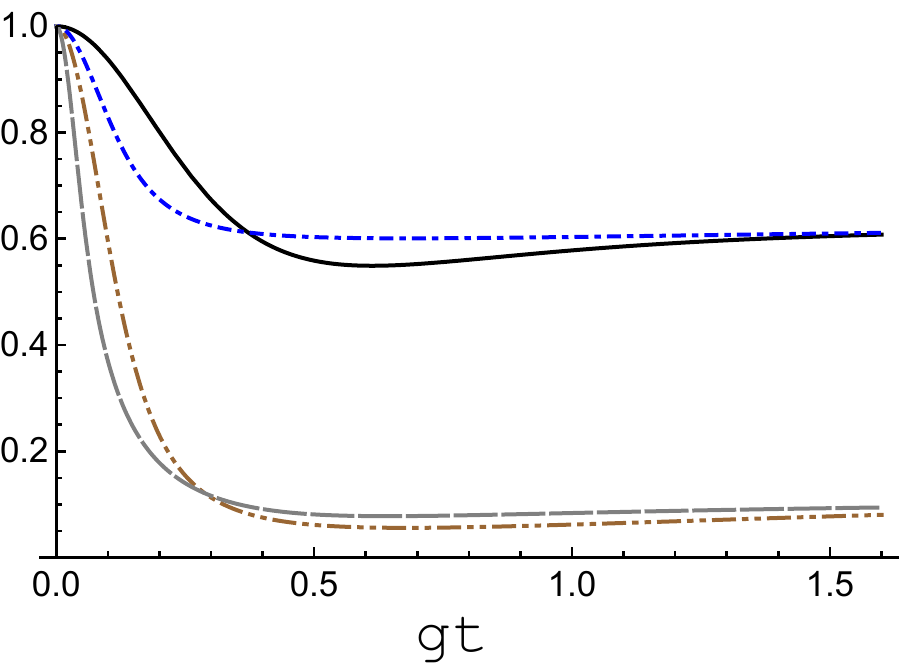}
  \label{fig:decofidd2n1}
    \centerline{(a) $N_{uno}=1$ }
\end{minipage}%
\begin{minipage}{.5\columnwidth}
  \centering
\includegraphics[width=0.9\textwidth]{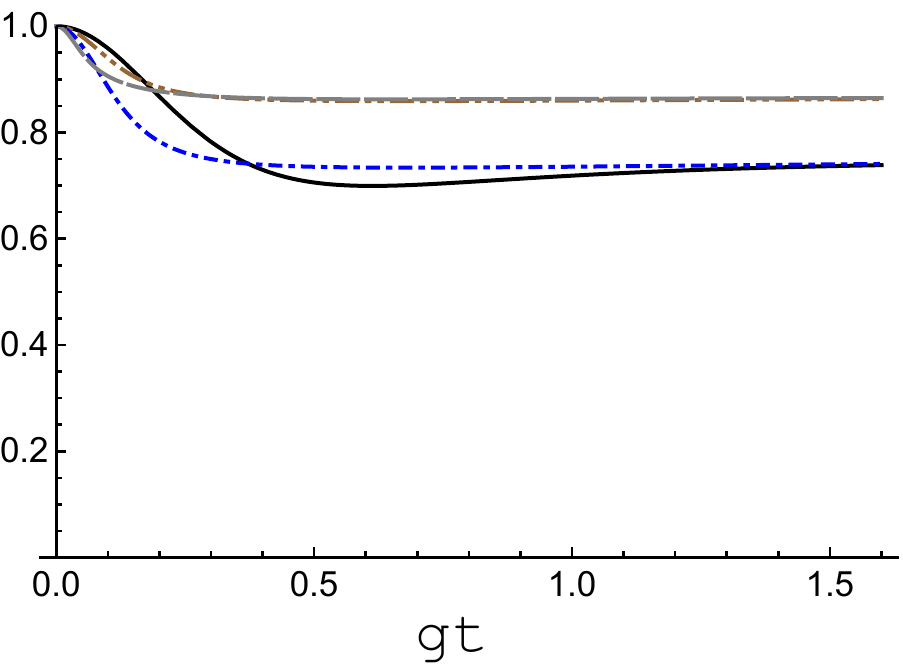}
  \label{fig:decofidd8n1}
    \centerline{(b)  $N_{mac}=1$ }
\end{minipage}\\
\vspace{.5cm}
\begin{minipage}{.5\columnwidth}
  \centering
\includegraphics[width=0.9\textwidth]{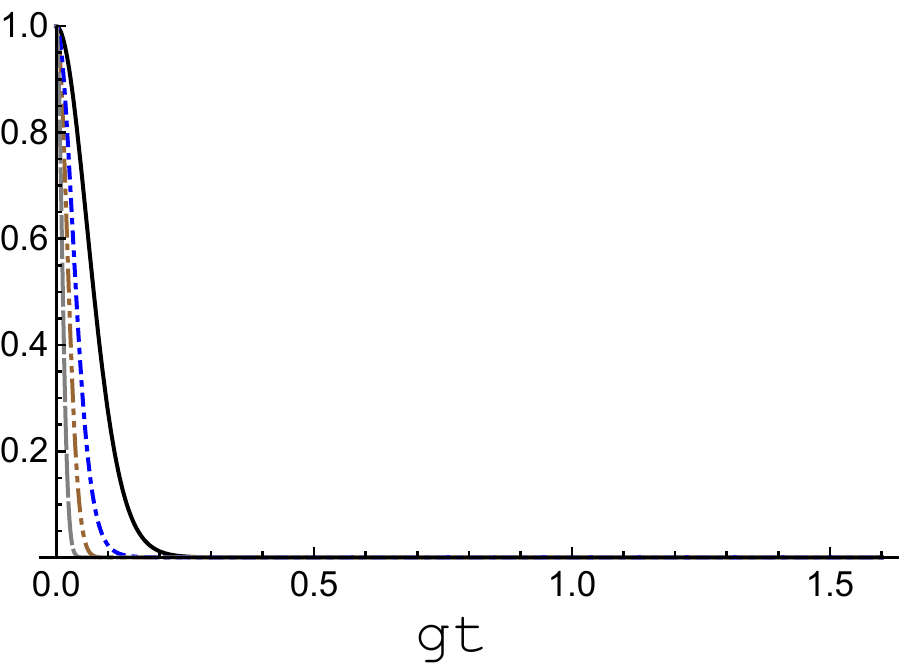}
  \label{fig:decofidd2n20}
    \centerline{(c)  $N_{uno}=20$ }
\end{minipage}%
\begin{minipage}{.5\columnwidth}
  \centering
\includegraphics[width=.9\textwidth]{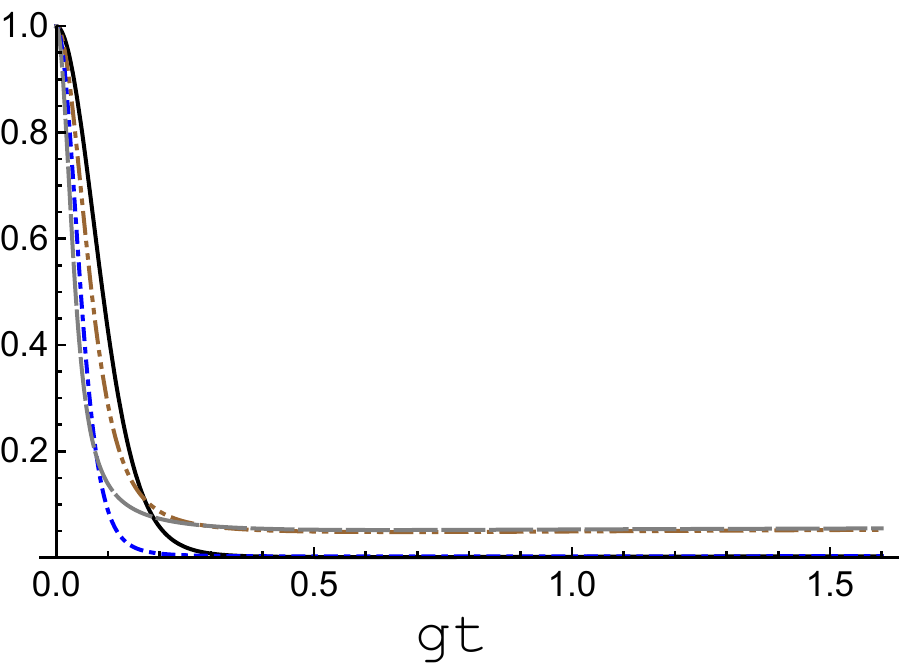}
  \label{fig:decofidd8n20}
    \centerline{(d) $N_{mac}=20$ }
\end{minipage}
\caption{(Color online). Time dependence of the exact full averages of the decoherence factor (a),(c) and the  super-fidelity (b),(d)
 for different dimensions and macrofraction sizes.
The  two-point correlation function $R_2(a,a')$ averages of the exact solutions (\ref{eq:GUEavg}, \ref{eq:GUEavgB}) were used, with $\<\tr \rho^2_0\>_{Bu}$ from \eqref{eq:HS}.
Different combinations of $\{d_S,d\}$ are plotted: $\{2,2\}$ solid black; $\{2,10\}$ brown dot-dot-dash; $\{10,2\}$ blue dot-dash;
$\{10,10\}$ gray long-dash. The time is in  the the units of the interaction strength $g\equiv 1/\sqrt{\eta_S\eta_E}$ \cite{nota2}.}
\label{fig:exact}
\end{figure}
%----------END FIGURE ----
%-------------------------
%-------------------------

Next crucial step is to use the result of \cite{MKH16} estimating an optimal trace distance between \eqref{ptr} and an ideal SBS state 
on the coarse-grained level of macrofractions:
\begin{eqnarray}
&&\epsilon_{SBS}(t)\equiv \frac{1}{2}||\rho_{S:E_{obs}}(t) - \rho_{SBS}||_{\tr}\label{e}\\
&&\leq \sum_a\sum_{a'\ne a} \left[|c_{aa'}| \sqrt{\Gamma_{aa'}^{uno}(t)}
+\sqrt{p_ap_{a'}}\sum_{mac}\sqrt{F^{mac}_{aa'}(t)}\right]. \nonumber
\end{eqnarray}
Using $p_a,|c_{aa'}|\leq1$, $\<\sqrt{f}\>\leq\sqrt{\< f\>}$ for $f\geq 0$, the super-fidelity bound,
and the Result \ref{result2}, 
estimation \eqref{e} gives:
\begin{rslt}\label{r3}
Averaged over all the von Neumann measurements \eqref{ham} and the initial conditions, the optimal distance of the actual state \eqref{ptr} to an ideal SBS state satisfies:
\be
\<\<\epsilon_{SBS}(t)\>\>\xrightarrow[t\gg\tau_{SBS}]{}O\left[d_S^2\left(\E^{-\frac{N_{uno}}{2}\log d}+\mathcal M\E^{-\frac{N_{mac}}{2d}}\right)\right]
\ee
where $\tau_{SBS}$ is the larger of \eqref{tfid} and $\mathcal M$ is the number of macrofractions into which the observed degrees of freedom of the apparatus are coarse-grained.
\end{rslt}

Finally,  since $0\leq \epsilon_{SBS}(t)\leq 1$ is a bounded random variable for any $t$, it follows from the Hoeffding inequality \cite{Hoe} that:
$P[|\epsilon_{SBS}(t)-\<\<\epsilon_{SBS}(t)\>\>|\geq \delta]\leq 2\E^{-2\delta^2 }$ for any $\delta\geq 0$.
This, together with Result \ref{r3} shows the genericity of the SBS formation for large enough apparatuses and long enough times.

{\bf Conclusions.--}A measurement is an inevitable part of any quantum experiment
and the results must inevitably be encoded into macroscopic degrees of freedom and become effectively classical 
for us to read. This in particular entails becoming objective. We studied this process using the general von Neumann 
measurement scheme (\ref{ham}) with a macroscopic measuring apparatus. %and recently introduced Spectrum Broadcast Structures as indicators of objectivity. 
A huge amount of degrees of freedom ($N \sim 10^{23}$) makes it in practice impossible to observe them all and to control each
individual coupling. A way to model this physical situation is to introduce some randomness and ask questions about genericity. We did it in two steps:
First we randomized the measurement device side  (the observables and the generically noisy initial states) and showed that after including the 
inevitable losses and macroscopic coarse-graining, a post measurement state approaches the so called SBS form asymptotically for almost
any initial conditions and couplings. The timescales of this process depended on the spectral gap of the measured observable on the system side.  
Afterwards, to get rid of this dependence, we went beyond a single experiment scenario, randomizing the measured observable too.
An interesting aspect of that second randomization is that
this may be viewed as a quite natural assumption of any quantum system
interacting with many objects. Indeed it is natural to assume
that it interacts with each of the objects with some fixed,
yet different than with the others, way. Since there are many
objects, then the averaging effect comes from that variety of
the interactions and can be viewed as a self-averaging of the
system plus environment complex. 
This led to our central result: Almost any quantum measurement produces  objective outcomes on the macroscopic level
on the timescale %of the order $(g\sqrt{N_{mac}d_S})^{-1}$, where $g$ is the effective coupling 
%strength, $d_S$ the measured system dimension, and $N_{mac}$ describes the apparatus coarse-graining.
given by the larger of \eqref{tfid}. 
This is a universal, model-independent result.

We believe one can go beyond the genericity notion used here (Hoeffding inequality)
and show the concentration of measure phenomenon, e.g. by combining
the results for the Wigner-type matrices \cite{Anderson} with the methods of \cite{Popescu-thermo,distinguishability}.
Another possible future direction is to go beyond the quantum measurement limit and consider 
non-trivial dynamics of the system and the measuring device. A candidate 
tool for such an analysis already exists in the form of dynamical SBS \cite{Tuziemski2015_objectivity}.

\section*{Acknowledgements}
This work was carried out at the National Quantum Information Centre in Gda\'nsk.
We thank J. Wehr, P. Grangier,  M. Horodecki, and R. Horodecki for discussions.
JKK and PH acknowledge the financial support of the John Templeton Foundation
through the grant ID \#56033. EAA is supported by the grant Sonata Bis 2014/14/E/ST2/00020 from National
Science Centre. P\'C thanks the Seventh framework programme EU grant
RAQUEL No 323970 and the grant PRELUDIUM 2015/17/N/ST2/04047 from National
Science Centre.

%Finally, we turn to the term with highest degree. This is achieved when $n=d-2$, $m=d-1$, $k=l=0$. Hence, the highest order term with its corresponding coefficient is:
%\be
%-\frac{2}{(d-2)!(d-1)!}\tilde{\Delta}^2(2d-3)
%\ee
\bibliographystyle{apsrev}

\newpage
\appendix
\widetext

\setcounter{figure}{0}
\renewcommand{\theequation}{A\arabic{equation}}
\renewcommand{\thefigure}{A\arabic{figure}}

\section{Ensemble average over the apparatus}
\subsection{Average over the Haar distributed unitary transformations U} \label{App:a}
In this Section we average the decoherence and the super-fidelity factors over the Haar measure. 
Due to the assumed independent identical distribution (i.i.d.) of the apparatus observables $B_k$, $k=1,\dots, N$, it is enough to calculate the averages over a single
observable only. This is what we shall calculate, neglecting for brevity the index $k$.
We start with the decoherence factor and prove that:
\begin{theorem}
The decoherence factor for the single copy of the environment average over the Haar distributed unitary transformations $U$ is equal to:
\be
\label{eq:HaarAvg2}
\< \Gamma_{aa'}(t) \>_{\cal U} = |\tr D|^2 \frac{d - \tr [\rho^2_0]}{d(d^2-1)} + \frac{d \tr [\rho^2_0] - 1}{d^2-1}
\ee
where $d$ is the local dimension of the environment.
\end{theorem}

We first write the decoherence factor as:
\be
\begin{split}
\label{eq:gamma_appendix}
\Gamma_{aa'}(t) &= \left|\tr [\E^{- i (a-a')B t} \rho_{0}]\right|^2 = \tr[\E^{- i (a-a')B t} \rho_{0}] \tr[\E^{- i (a-a')B t} \rho_{0}]^{\dagger} = \tr[\E^{- i (a-a')B t} \rho_{0}] \tr[(\E^{ -i (a-a')Bt} )^{\dagger} \rho_{0}] \\
&= \tr [U DU^{\dagger} \rho_0 \ot U D^{\dagger} U^{\dagger} \rho_0] = \tr [(U^\hc \ot U^\hc) (\rho_0 \ot \rho_0) (U \ot U) (D \ot D^{\dagger})]
\end{split}
\ee
where we diagonalized the observable $B$ as $B=U\textrm{diag}[\lambda_1,\dots,\lambda_d]U^\dagger$ and defined: 
\be\label{D}
D\equiv \text{diag}\left[\E^{-i \Delta_t \lambda_1},\dots,\E^{-i \Delta_t \lambda_d}\right],\quad \Delta_t\equiv(a-a')t. 
\ee
We also used $\tr A\tr B=\tr(A\ot B)$ in the second line and the following fact in the first step:
\begin{fact}
\label{trabab}
For any operator $X$ the following is true
\be
|\tr X|^2 = \tr X \tr X^{\dagger},
\ee
where $ \dagger $ stands for hermitian conjugation.
\begin{proof}
\be
|\tr X|^2 = \tr X \overline{\tr X} = \tr X \overline{\tr X^T} = \tr X \tr X^{\dagger},
\ee
where $\overline{X}$ stands for the complex conjugation of $X$ and we used $\tr X = \tr X^T$, and $\overline{X^T} = X^{\dagger}$.
\end{proof}
\end{fact}

We will also need two more well known facts:
\begin{fact}
\label{eq:aTrVAB}
For any operators $A$, $B$ and the $SWAP$ operator $\mathbb V$, we have that:
\be \label{tr}
\tr [\mathbb \mbV A \otimes B] = \tr (A B).
\ee
\end{fact}
\begin{proof}
Let us write the $SWAP$ operator as:
\be \mathbb V = \sum_{ij} |ij\> \< ji| \label{Sb} \ee
Inserting \eqref{Sb} into \eqref{tr} we have that:
\ben \tr [\mathbb V A \otimes B] &=& \tr \left( \sum_{ij} |ij\> \< ji| A \ot B \right) = \tr \left( \sum_{ij} \< ji| A \ot B |ij\> \right) \nonumber \\
&=&  \sum_{ijkl} \< ji| A |kl\> \< kl| B |ij\>  = \tr (A B).
\een
\end{proof}

%Next, we will need the well-know result from the group-theory concerning the result of the unitary twirling:
\begin{fact}
\label{Haar_int}
For any hermitian operator $X$ from $\mathbb C^d$ to $\mathbb C^d$, it holds: 
\be
\begin{split}
\int{\rm d}{\cal U} U \ot U X \ot X U^{\dagger} \ot U^{\dagger} 
%= \int{\rm d}{\cal \hat{U}} \hat{U}^{\dagger} \ot \hat{U}^{\dagger} X \ot X \hat{U} \ot \hat{U}
= \frac{2}{d(d+1)} \tr[\Pi_{\text{sym}}X \ot X]\Pi_{\text{sym}} + \frac{2}{d(d-1)}\tr[\Pi_{\text{asym}}X \ot X]\Pi_{\text{asym}},
\end{split}
\ee
where $\Pi_{\text{sym}}$ and $\Pi_{\text{asym}}$ are the orthogonal projectors onto the symmetric and antisymmetric subspaces, respectively, equal to
\be
\Pi_{\text{sym}} \equiv \frac{\mathbb I + \mathbb V}{2}, \quad \Pi_{\text{asym}} \equiv \frac{\mathbb I - \mathbb V}{2},
\ee
where $\mathbb V$ is the $SWAP$ operator.
\end{fact}

We now integrate Eq. \eqref{eq:gamma_appendix} over $U \ot U$. Using linearity of the trace we pull the integral inside the trace:
\be
\begin{split}
\< \Gamma_{aa'}(t) \>_{\cal U} &= \int{\rm d}{\cal U} \tr (U^{\dagger} \ot U^{\dagger} \rho_0 \ot \rho_0 U \ot U D \ot D^{\dagger}) 
= \tr \left[ \left( \int{\rm d}{\cal U}U \ot U \rho_0 \ot \rho_0 U^{\dagger} \ot U^{\dagger} \right) D \ot D^{\dagger} \right] \\
\end{split}
\ee
We then use Fact \ref{Haar_int} with $X \equiv \rho_0$. We can easily calculate 
$\tr[\Pi_{\text{sym}}\rho_0 \ot \rho_0]$ and $\tr[\Pi_{\text{asym}}\rho_0 \ot \rho_0]$ 
using Facts \ref{trabab} and \ref{eq:aTrVAB} and obtain:
\be
\label{eq:tr_sym_1}
\tr[\Pi_{\text{sym}}\rho_0 \ot \rho_0]\Pi_{\text{sym}} = \frac{1}{2} \tr \left[ (\mathbb I + \mathbb V)(\rho_0 \ot \rho_0) \right]\frac{\mathbb I + \mathbb V}{2} = \frac{1+\tr[{\rho_0}^2]}{2}\frac{\mathbb I + \mathbb V}{2},
\ee
and for the antisymmetric projector
\be
\label{eq:tr_asym_1}
\tr[\Pi_{\text{asym}}\rho_0 \ot \rho_0]\Pi_{\text{asym}} = \frac{1}{2} \tr \left[ (\mathbb I - \mathbb V)(\rho_0 \ot \rho_0) \right]\frac{\mathbb I - \mathbb V}{2} = \frac{1-\tr[{\rho_0}^2]}{2}\frac{\mathbb I - \mathbb V}{2}.
\ee
We then again use Facts \ref{trabab} and \ref{eq:aTrVAB} to calculate the remaining traces $\tr \left[(\mathbb I \pm \mathbb V) D \ot D^{\dagger}\right]$, keeping in mind that $D$ is hermitian and that $\tr D^2= d$. This finally gives:
\be
\begin{split}
\< \Gamma_{aa'}(t) \>_{\cal U} = |\tr \hat{D}|^2 \frac{d - \tr [\rho^2_0]}{d(d^2-1)} + \frac{d \tr [\rho^2_0] - 1}{d^2-1},
\end{split}
\ee
proving our Theorem. \qed

Using the same technique, one can also calculate the average of the super-fidelity factor. The only non-trivial part is the Hilbert-Schmidt product between
the apparatus states $\rho_a(t)$ and $\rho_{a'}(t)$. Using the same notation as in Eq. \eqref{eq:gamma_appendix} we obtain:
\be
\begin{split}
\label{Sfid}
\tr\left(\rho_a(t)\rho_{a'}(t)\right) &= \tr [\rho_0\E^{i (a-a')B t} \rho_{0}\E^{-i (a-a')B t}] = 
\tr [U^{\dagger} \rho_0 UD^\dagger U^{\dagger} \rho_0 U D] = \tr [\mbV (U^\hc \ot U^\hc) (\rho_0 \ot \rho_0) (U \ot U) (D^\hc \ot D)],
\end{split}
\ee 
where in the last step we used Fact \ref{eq:aTrVAB}. We note that the only difference between the above Hilbert-Schmidt factor and the decoherence factor 
\eqref{eq:gamma_appendix} is the presence of the SWAP operator $\mbV$. Repeating the same steps as above gives:
\be
\<G_{aa'}(t)\>_{\cal U} =S_{lin}(\rho_0)+\frac{d - \tr\rho^2_0}{d^2-1} + |\tr D|^2 \frac{d\tr \rho^2_0-1}{d(d^2-1)}.
\ee

Finally, we evaluate $|\tr D|^2$ from its definition in Eq. \eqref{D}:
\begin{eqnarray}\label{trD}
|\tr D|^2 \equiv d + 2f_t(\pmb a,\pmb{\lambda}), \quad f_t(\pmb a,\pmb{\lambda}) \equiv \sum_m\sum_{n>m} \cos\left[\Delta_t(\lambda_n-\lambda_m)\right] = \sum_m\sum_{n>m}\cos\left[(a-a')(\lambda_n-\lambda_m)t\right] ,
\end{eqnarray}
which is a function of the eigenvalues $\pmb a$ of the observable $A$ and the eigenvalues $\pmb\lambda$ of $B$.

\subsection{Averaging over the eigenvalues} \label{App:b}
After averaging over the unitary group in Sec. \ref{App:a}, we perform the average over the GUE eigenvalue distribution:
\be\label{Pgue}
P_{gue}(\pmb{\lambda}) = \frac{1}{Z} e^{- \frac{1}{2}\eta_E \sum_m \lambda_m^2} \prod_{i<j} |\lambda_j - \lambda_i|^2.
\ee
Here $\eta_E$ is the eigenvalue scale of the observable $B$  and $Z$ is a normalization constant (the GUE partition function). The task then is to find the following average:
\be
\<f_t(\pmb a,\pmb{\lambda})\> = \< \sum_i\sum_{j>i} \cos[\Delta_t(\lambda_i - \lambda_j)] \> = 
\sum_i\sum_{j>i} \int d\pmb{\lambda} \SP P_{gue}(\pmb{\lambda}) \cos[\Delta_t(\lambda_i - \lambda_j)]
\ee
Quite surprisingly, this average can be performed explicitly using the standard methods of dealing with GUE \cite{Mehta_book}.
We first introduce the harmonic oscillator wave functions:
\be\label{phi}
\phi_n (x) \equiv \frac{1}{\sqrt{\sqrt{2\pi} n!}}e^{-\frac{x^2}{4}}He_n(x).
\ee
Notice that we define the wave functions using the so-called "probabilist" Hermite polynomials:
\[
He_n(x) \equiv (-1)^n e^{\frac{1}{2}x^2} \frac{d^n}{dx^n} e^{-\frac{1}{2}x^2} 
\]
That is, they are orthogonal with respect to the weight function $\exp[-x^2/2]$, and are related to the physicist's polynomials $H_n(x)$ via $He_n(x)=2^{-n/2}H_n(x/\sqrt{2})$. 
Of course we still have $\int dx \SP \phi_n(x)\phi_m(x) = \delta_{nm}$. Then the GUE eigenvalue distribution takes on 
a very compact and elegant form, after rescaling $\lambda_k \equiv \frac{\zeta_k}{\sqrt{\eta_E}}$  \cite{Mehta_book} (6.2.4):
\be\label{det}
P_{gue}(\pmb{\lambda})d\pmb{\lambda} = \frac{1}{d!} \det[\phi_{j-1}(\zeta_i)]^2 d\pmb{\zeta}, 
\ee
where $i,j=1,\dots, d$.

\subsubsection{Exploiting the symmetry}
A crucial step is the realization that this average has an index permutation symmetry. Let $\sigma \in S_d$, be a permutation, then:
\be
 P_{gue}(\lambda_{\sigma(1)},\dots,\lambda_{\sigma(d)}) = P_{gue}(\lambda_1,\dots,\lambda_d).
\ee
Analogously we have (keeping the eigenvalues $\pmb a$ fixed):
\be
f_t(\pmb a,\lambda_{\sigma(1)},\dots,\lambda_{\sigma(d)})=\sum_{i<j} \cos[\Delta_t(\lambda_{\sigma(i)} - \lambda_{\sigma(j)})]=\sum_{i<j} \cos[\Delta_t(\lambda_i - \lambda_j)] = f_t(\pmb a,\pmb\lambda).
\ee
This is because in both expressions the eigenvalue functions are symmetric and all pairs of indices are taken (i.e. the product or sum is over all $i<j$). Equivalently, we can recall that 
$\tr D= \sum_i \E^{-\ii \Delta_t \lambda_i}$ and from Eq. \eqref{trD} $f_t(\pmb a,\pmb\lambda)=1/2(|\tr D|^2 -d)$, which is clearly symmetric under the permutations. 
Hence, the calculation of the average $f_t(\pmb a,\pmb{\lambda})$ reduces to a single term:
\be
\<f_t(\pmb a,\pmb{\lambda})\> = \frac{d(d-1)}{2} \int d\pmb{\lambda} \SP P_{gue}(\pmb{\lambda}) \cos[\Delta_t(\lambda_1 - \lambda_2)]
\ee
As the integrand only depends on two variables, we can take the marginal distribution, which is essentially the 2-point correlation function, defined as 
\cite{Mehta_book} (6.1.2):
\be\label{R2}
R_2(\lambda_1,\lambda_2) \equiv \frac{d!}{(d-2)!} \int d\lambda_3 \cdots d\lambda_d \SP P_{gue}(\lambda_1,\dots,\lambda_d).
\ee
Hence, we have reduced the problem to the integral:
\be
\<f_t(\pmb a,\pmb{\lambda})\> = \frac{1}{2} \int d\lambda_1 d\lambda_2 \SP R_2(\lambda_1,\lambda_2) \cos[\Delta_t(\lambda_1 - \lambda_2)].
\label{eq:fintegral}
\ee

\subsubsection{The crucial integral}
To calculate \eqref{eq:fintegral} we will need the following integral:
\be
J_{n,m}(\alpha) \equiv \int dx \SP \phi_n(x)\phi_m(x) e^{\ii \alpha x}.
\label{eq:Jnm}
\ee
In fact \eqref{eq:Jnm} may be interpreted as a special case of a matrix element of the displacement operator 
$D(\beta)=\exp[\beta \hat{a}^\dagger - \beta^*\hat{a}]$ in the Fock basis $\{|n\ra\}_{n\in\mathbb{N}}$. We recover our integral setting  
$\beta \equiv\ii \alpha$, $\alpha \in \mathbb{R}$. This turns out to be a well known quantity in quantum optics (see e.g. \cite{Perelomov}), but for completeness we present its calculation below. 
Without a loss of generality, we will assume $m\geq n$. First, we express the wavefunctions through the (probabilist) Hermite polynomials as in Eq. \eqref{phi} and
use the generating function for $He_n(x)$ with parameters $r,s$ to perform the integral:
\begin{align}
J_{n,m}(\alpha)&= \frac{1}{\sqrt{2\pi n! m!}} \int dx \SP He_n(x) He_m(x) e^{-\frac{1}{2}x^2 + \ii \alpha x} \\
 &= \frac{1}{\sqrt{2\pi n! m!}} \frac{\pd^n }{\pd r^n}\biggr\rvert_{r=0} \frac{\pd^m }{\pd s^m}\biggr\rvert_{s=0} \int dx \SP e^{(\ii \alpha + r + s)x - \frac{1}{2} (r^2 + s^2 + x^2)} \\
 &= \frac{e^{-\frac{1}{2} \alpha^2}}{\sqrt{ n! m!}}\frac{\pd^n }{\pd r^n}\biggr\rvert_{r=0} \frac{\pd^m }{\pd s^m}\biggr\rvert_{s=0} e^{\ii \alpha r + \ii \alpha s + rs} \\
 &=  \frac{e^{-\frac{1}{2} \alpha^2}}{\sqrt{ n! m!}}\frac{\pd^n }{\pd r^n}\biggr\rvert_{r=0}  (\ii \alpha + r)^m e^{\ii \alpha r}
\end{align}
Then we use the binomial formula for the derivatives $\frac{d^n}{dx^n} f(x)g(x) = \sum_{k=0}^n \binom{n}{k} \frac{d^{n-k}f(x)}{dx^{n-k}}\frac{d^{k}g(x)}{dx^k}$. 
Since $m\geq n$ we don't run into any unexpected problems and obtain:
\be
J_{n,m}(\alpha) = \frac{e^{-\frac{1}{2} \alpha^2}}{\sqrt{ n! m!}} \sum_{k=0}^n \binom{n}{k}\binom{m}{k}k! (\ii \alpha)^{n+m-2k}
\label{eq:Jpoly}
\ee
This may be nicely expressed in terms of the associated Laguerre polynomials as (taking $\alpha \in \mathbb{R}$):
\be
J_{n,m}(\alpha) = e^{-\frac{1}{2} \alpha^2} \sqrt{ \frac{n!}{m!}} (\ii \alpha)^{m-n} L_{n}^{(m-n)}(\alpha^2),
\label{eq:JLaguerre}
\ee
where:
\be
L_n^{(m)} (x) \equiv \sum_{k=0}^{n} \binom{n+m}{n-k} \frac{(-x)^k}{k!}
\ee
(we adopt the common standardization for the Laguerre polynomials that the leading coefficient is equal to $(-1)^n/n!$). 
Eq. \eqref{eq:Jnm} can also be expressed more compactly in terms of the, so-called, 2D Laguerre functions introduced in \cite{WunscheLaguerre}: 
\be
\la m | D(\alpha) | n\ra = (-1)^n \sqrt{\pi} l_{m,n}(\alpha,\alpha^*)
\ee
for a general complex displacement $\alpha$. The 2D Laguerre functions are defined as \cite{WunscheLaguerre}:
\be
l_{m,n} (z,z^*) \equiv \frac{1}{\sqrt{\pi}} e^{-\frac{zz^*}{2}} \frac{1}{\sqrt{m!n!}}\sum_{j=0}^m \binom{m}{j}\binom{n}{j}j! (-1)^jz^{m-j}z^{*n-j}.
\ee

\subsubsection{Putting the results together}
We return to calculating the integral \eqref{eq:fintegral}. We use Eq. \eqref{det}, rescale the variables, and 
introduce a more friendly notation $(x,y)\equiv(\zeta_1,\zeta_2)$:
\be
\<f_t(\pmb a,\pmb{\lambda})\> = \frac{1}{2} \int dx dy \SP R_2(x,y) \cos[\tilde{\Delta}_t(x-y)],
\ee
where:
\be
\tilde{\Delta}_t \equiv \frac{(a-a')t}{\sqrt{\eta_E}}.
\ee
Now, Dyson's Theorem will let us calculate the 2-point correlation function \cite{Mehta_book} Thm 5.14 , (6.2.6-7):
\be\label{R21}
R_2(x,y) = K(x,x)K(y,y)-K(x,y)^2,
\ee
where the kernel is defined through the oscillator wave-functions \eqref{phi} as:
\be\label{K}
K(x,y) \equiv \sum_{j=0}^{d-1} \phi_j(x) \phi_j(y).
\ee
Hence, we can express our integral as the following sum:
\be
\<f_t(\pmb a,\pmb{\lambda})\> = \sum_{n=0}^{d-1}\sum_{m=0}^{d-1} \int dxdy \SP \left(\phi_n(x)^2 \phi_m(y)^2 - \phi_n(x)\phi_m(x)\phi_n(y)\phi_m(y) \right) \cos[\tilde{\Delta}_t(x-y)].
\ee
By expressing the cosine function in exponential form, the integrals become separable and we obtain:
\be
\<f_t(\pmb a,\pmb{\lambda})\> = \frac{1}{4} \sum_{n=0}^{d-1}\sum_{m=0}^{d-1} \left(A_{n,m} + \tilde{A}_{n,m} - 2 B_{n,m} \right),
\ee
where we introduced auxiliary functions:
\begin{align}
A_{n,m}&\equiv J_{n,n}(\tilde{\Delta}_t)J_{m,m}(-\tilde{\Delta}_t), \\
\tilde{A}_{n,m}&\equiv J_{n,n}(-\tilde{\Delta}_t)J_{m,m}(\tilde{\Delta}_t), \\
B_{n,m}&\equiv  J_{n,m}(\tilde{\Delta}_t)J_{n,m}(-\tilde{\Delta}_t).
\end{align}
Now, we are able to separate the $n,m$ summation into three parts $n<m$, $n=m$, and $n>m$. From the definition of the auxiliary functions, we easily see that the diagonal summation $n=m$ vanishes. The remaining sums $n<m$ and $n>m$ become the same, since $A_{m,n}=\tilde{A}_{n,m}$ and $B_{m,n}=B_{n,m}$. Hence, for convenience's sake we will calculate the sum $m>n$ only. We use the explicit result \eqref{eq:Jpoly} for the $J_{n,m}$ and obtain:
\be
A_{n,m} = e^{-\tilde{\Delta}_t^2} \sum_{k=0}^{n} \sum_{l=0}^m \binom{n}{k} \binom{m}{l} \frac{(-1)^{n+m-k-l}\tilde{\Delta}_t^{2(n+m-k-l)}}{(n-k)!(m-l)!}.
\ee
And by doing the same for $\tilde{A}$, we indeed realize that $A_{n,m}=\tilde{A}_{n,m}$. For $B_{n,m}$, in turn, we obtain:
\be
B_{n,m} =e^{-\tilde{\Delta}_t^2} \sum_{k=0}^{n} \sum_{l=0}^n \binom{n}{k} \binom{m}{l} \frac{(-1)^{k+l}\tilde{\Delta}_t^{2(n+m-k-l)}}{(m-k)!(n-l)!}.
\ee
Hence, putting it all together (with a factor of 2, since now we only sum $m>n$), we arrive at:
\be\label{av}
\<f_t(\pmb a,\pmb{\lambda})\> =e^{-\tilde{\Delta}_t^2} \sum_{n=0}^{d-2}\sum_{m=n+1}^{d-1} \left[ \sum_{k=0}^{n} \sum_{l=0}^m \binom{n}{k} \binom{m}{l} \frac{(-1)^{n+m-k-l}\tilde{\Delta}_t^{2(n+m-k-l)}}{(n-k)!(m-l)!} - \sum_{k=0}^{n} \sum_{l=0}^n \binom{n}{k} \binom{m}{l} \frac{(-1)^{k+l}\tilde{\Delta}_t^{2(n+m-k-l)}}{(m-k)!(n-l)!}\right].
\ee
This result may also be rewritten using the associated Laguerre polynomials and Eq. \eqref{eq:JLaguerre}:
\be
\<f_t(\pmb a,\pmb{\lambda})\> =e^{-\tilde{\Delta}_t^2} \sum_{n=0}^{d-2}\sum_{m=n+1}^{d-1} \left[ L_n^{(0)}(\tilde{\Delta}_t^2)L_m^{(0)}(\tilde{\Delta}_t^2) - \frac{n!}{m!}\tilde{\Delta}_t^{2(m-n)} [L_n^{(m-n)}(\tilde{\Delta}_t^2) ]^2 \right]
\label{eq:fLaguerre}
\ee
One of the sums can be performed using the following identity for the associated Laguerre polynomials:
\be
\sum_{m=0}^M L^{(\alpha)}_m(x)=L^{(\alpha+1)}_M(x),
\ee
giving:
\be
\<f_t(\pmb a,\pmb{\lambda})\> =e^{-\tilde{\Delta}_t^2} \left[L_{d-1}^{(1)}( \tilde{\Delta}_t^2)L_{d-2}^{(1)}( \tilde{\Delta}_t^2)  - \sum_{n=0}^{d-2} 
L_{n}^{(0)}( \tilde{\Delta}_t^2)L_{n}^{(1)}( \tilde{\Delta}_t^2) 
- \sum_{n=0}^{d-2}\sum_{m=n+1}^{d-1}  \frac{n!}{m!}\tilde{\Delta}_t^{2(m-n)} [L_n^{(m-n)}(\tilde{\Delta}_t^2)]^2 \right]. 
\ee
However, we keep \eqref{eq:fLaguerre} in the main text since it is more compact.

Please note that in the main text we do not directly use $\<f_t(\pmb a,\pmb{\lambda})\>$, but rather separate its Gaussian and polynomial parts, i.e.\ 
\be\label{pd}
\<f_t(\pmb a,\pmb{\lambda})\>\equiv\E^{-\tilde{\Delta}_t^2}p(d,\tilde{\Delta}_t).
\ee

\renewcommand{\theequation}{B\arabic{equation}}
\renewcommand{\thefigure}{B\arabic{figure}}

\section{Short time analysis}\label{shorttime}
In this Section we perform the final averaging over the system observable $A$. In the previous Section, we have used the i.i.d. property
of the apparatus ensemble to reduce the big, compound averages to single copy ones. Here, however, we cannot do so, as ultimately we are interested in the 
macroscopic quantities $\Gamma_{aa'}^{N_{uno}}$, $G_{aa'}^{mac}$. %More precisely, for the reasons that will become clear later,  
%we need to calculate the averages of the sums  $\sum_a\sum_{a'\ne a}\Gamma_{aa'}^{N_{uno}},G_{aa'}^{mac}$. 
Thus we need: 
\be\label{sums}
\<\< X^{f}(t)\>\>_{aa'}\equiv \int d\pmb a P_{gue}(\pmb a) \<X_{aa'}(t)\>^{N_f},
\ee
where $X_{aa'}=\Gamma_{aa'}$ or $G_{aa'}$ and $f=uno$ or $mac$ respectively. We note that both 
$\<\Gamma_{aa'}(t)\>$, $\<G_{aa'}(t)\>$ depend on $A$ only through the eigenvalue differences  $|a-a'|$. 
Thus, the $A$-averaging reduces to averaging over the eigenvalues only with its own GUE eigenvalue distribution:
\be\label{PS}
P_{gue}(\pmb a) \equiv P_{gue}(a_1, \dots,a_{d_S})\equiv \frac{1}{Z_S} e^{- \frac{1}{2}\eta_S \sum_l a_l^2} \prod_{i<j} |a_j - a_i|^2,
\ee
where $d_S$ is the system dimension and $\eta_S$ is the eigenvalue scale of the system observable $A$. From the permutational symmetry 
of the GUE distribution  (best seen through the Vandermonde determinant), the integral in Eq. \eqref{sums} reduces to the integration with the same 2-point correlation function \eqref{R2}, but now defined for the distribution
\eqref{PS} and thus all the averages for different pairs $aa'$ are the same and equal to:
\be\label{finalav}
\<\< X^{f}(t)\>\>=\left[\frac{d_S!}{(d_S-2)!}\right]^{-1} \int dada' R_2(a,a') \<X_{aa'}(t)\>^{N_f}.
\ee 
The resulting integral is too complicated to be performed analytically and actually this is in fact not needed as we see from  Eq. \eqref{pd} that $\<f_t(\pmb a,\pmb\lambda)\>$
will eventually decay so that both factors will approach their noise-floor values (cf. Result 1 from the main text). What we are interested in are the relevant timescales.
We can estimate lower bounds on those timescales from the decay times of \eqref{pd}.
 We will perform this analysis in the following steps: o) assume a short-time limit;
i) approximate the polynomial $p(d,\tilde{\Delta})$ of Eq. \eqref{pd}; ii) approximate the $N_f$ power; iii) using Eqs. \eqref{R21},\eqref{K}, and \eqref{phi}
estimate the fastest decaying term in \eqref{finalav}. This will then give lower bounds on the desired times of the asymptotic approach: The latter are for sure greater than the initial decay times.

First, we assume $\tilde \Delta_t \ll1$, or $t\ll \sqrt{\eta_E}/|a-a'|$. 
The maximum of $|a-a'|$ is of the order of $\sqrt{d_S/\eta_S}$ from the Wigner semi-circle law, defining the short-time limit:
\be\label{stlim}
t\ll \sqrt{\frac{\eta_E\eta_S}{ d_S}}\equiv \frac{1}{g\sqrt{d_S}},
\ee
where $g\equiv 1/\sqrt{\eta_E\eta_S}$ is the effective interaction strength.
 
We now explicitly calculate the coefficients of the lowest order terms in $p(d,\tilde{\Delta})$ directly from Eq. \eqref{av}. One immediately sees that the
polynomial is even so the lowest terms are the constant and the quadratic ones.
The constant term occurs when $k+l=m+n$, but this can only occur in the first summand of the polynomial when $k=n$ and $l=m$, so that we have:
\be
\sum_{n<m} [1] = \frac{d(d-1)}{2}
\ee
The quadratic term occurs when the indices fulfill the condition $k+l+1=m+n$. On the first term, this can occur if ($k=n$, $l=m-1$) or if ($k=n-1$, $l=m$). 
On the second term this can occur only when the $m$ index is $m+1$ and the inner indices are $k=l=n$. Thus we obtain:
\be
\sum_{n<m}\left[m(-1)\tilde{\Delta}_t^2\right] + \sum_{n<m}\left[n(-1)\tilde{\Delta}_t^2\right] - \sum_n\left[(n+1)\tilde{\Delta}_t^2\right] = -\frac{d^2(d-1)}{2}\tilde{\Delta}_t^2
\ee
Thus, for short times we have:
\be\label{pshort}
p(d,\tilde{\Delta}_t) = \frac{d(d-1)}{2}\left( 1 - d \tilde{\Delta}_t^2 \right) + O(\tilde{\Delta}_t^4).
\ee

To proceed further, we upper bound the above expression by the Gaussian function: $1-d\tilde\Delta_t^2\leq\E^{-d\tilde\Delta_t^2}$ resulting from Eq. \eqref{pd} in the short-time bound:
\be\label{ft}
\<f_t(\pmb a,\pmb\lambda)\> \lesssim \E^{-(d+1)\tilde\Delta_t^2}\equiv\E^{-(t/\tau_{aa'})^2},\quad 
\tau_{aa'}\equiv \frac{\sqrt{\eta_E}}{\sqrt{d+1}|a-a'|}.
\ee
We then use the following 
approximation of the power:
\be
\left( \alpha + \beta \E^{-x} \right)^{N_f} \approx (\alpha + \beta)^{N_f} \E^{-N_f\frac{\beta}{\alpha + \beta}x}
\ee
for $x\ll 1$. We apply it to the single-copy averaged factors:
\begin{eqnarray}
&&\< \Gamma_{aa'}(t) \>=  \frac{1+ \tr \rho^2_0 }{d+1}+\< f_t(\pmb a,\pmb{\lambda})\> \frac{2(d - \tr \rho^2_0)}{d(d^2-1)}, \\
&&\< G_{aa'}(t) \>=S_{lin}( \rho_0)+ \frac{1+ \tr \rho^2_0 }{d+1}+\< f_t(\pmb a,\pmb{\lambda})\>  \frac{2(d\tr \rho^2_0-1)}{d(d^2-1)}, 
\end{eqnarray}
identifying $\beta$ with the constant terms and $\alpha$ with the multiplicative one. After a simple algebra, this leads to the following short time approximations
of the macroscopic factors: 
\begin{eqnarray}
&&\< \Gamma^{uno}_{aa'}(t) \>=\<\Gamma_{aa'}(t)\>^{N_{uno}} \approx \exp\left[- N_{uno} \SP \tilde{\Delta}_t^2 \left(d - \<  \tr \rho_0^2 \> \right)\right],\label{eq:ShortTimeGamma} \\
&&\< G^{mac}_{aa'}(t) \>=\< G_{aa'}(t) \>^{N_{mac}}\approx \exp\left[- N_{mac} \SP \tilde{\Delta}_t^2 \left(d \<  \tr \rho_0^2 \> -1 \right)\right]. 
\label{eq:ShortTimeFidelity}
\end{eqnarray}
For $\<  \tr \rho_0^2 \>$ we can use either the Hilbert-Schmidt or Bures measure from the main text. It is interesting to note that in any case $\<  \tr \rho_0^2 \> \propto 1/d$ , so that \eqref{eq:ShortTimeGamma} shows a dependence on both $N_{uno}$ and $d$ while \eqref{eq:ShortTimeFidelity} shows a dependence mainly on $N_{mac}$.

We are now ready to estimate the lower bound on the decay time of the integral \eqref{finalav}. Substituting Eq. \eqref{eq:ShortTimeGamma} or \eqref{eq:ShortTimeFidelity} 
and using Eqs. \eqref{R21}, \eqref{K}, and \eqref{phi}, we are left with a sum of integrals of the following structure:
\be\label{basic}
\int dada' He_m(a) He_n(a) He_j(a')He_l(a')\E^{-\mu_t(a-a')^2-\frac{1}{2}(a^2+a'^2)} \sim \int dada' a^r a'^s \E^{-\mu_t(a-a')^2-\frac{1}{2}(a^2+a'^2)}, 
\ee
where the powers satisfy $0\leq r,s\leq 2(d_S-1)$ (cf. Eq. \eqref{K}) and $\mu_t\equiv N_{uno}(gt)^2\left(d - \<  \tr \rho_0^2 \> \right)$ for the decoherence factor and 
$\mu_t\equiv N_{mac}(gt)^2 \left( d \<  \tr \rho_0^2 \> -1 \right)$ for the super-fidelity (cf. Eqs. (\ref{eq:ShortTimeGamma},\ref{eq:ShortTimeFidelity})).
To separate the integrals, we diagonalize the quadratic form in the exponent. Its eigenvalues read $(1+4\mu_t)/2$ and $1/2$. Denoting by $O$ the diagonalizing $SO(2)$
transformation, we pass to the new variables $(x,y)^T\equiv O(a,a')^T$ in which the integrals \eqref{basic} separate: 
\be
 \int dada' a^r a'^s \E^{-\mu_t(a-a')^2-\frac{1}{2}(a^2+a'^2)}\sim \left(\int dx x^\gamma \E^{-\frac{1+4\mu_t}{2}x^2}\right)
\left(\int dy y^\delta \E^{-\frac{1}{2}y^2}\right),
\ee
where $0\leq \gamma,\delta\leq 2(d_S-1)$, $\gamma+\delta=r+s$. The integrals on the right hand side are now elementary and 
depend on time as $1/(1+4\mu_t)^{\gamma+\epsilon}$, where $\epsilon=1$ or $1/2$ depending on the parity of the power $\gamma$.
Putting all together, the integral \eqref{finalav} has the following time dependence (we neglect possible time dependent coefficients, 
originating from the matrix $O$, as they are given by the sine and cosine functions and hence $\sim O(1)$):
\be
\<\< X^{f}(t)\>\>\sim \frac{1}{\sqrt{1+4\mu_t}}+\frac{1}{1+4\mu_t}+\frac{1}{(1+4\mu_t)^{3/2}}
+\cdots +\frac{1}{(1+4\mu_t)^{2(d_S-1)+\epsilon}}.
\ee
The fastest decaying is the last term, which for $\mu_t\ll 1$ or $t\ll O\left(\frac{1}{\sqrt{g^2N_f}}\right)$ can be approximated by an exponential:
\be
\frac{1}{(1+4\mu_t)^{2(d_S-1)+\epsilon}}\approx\textrm{exp}\left[-8d_S c(d)  N_f (gt)^2\right],
\ee
where we neglected the factor $-2+\epsilon=-1$ or $-3/2$ in the power as compared to $2d_S$ and $c(d)\equiv d - \<  \tr \rho_0^2 \>$
for the decoherence factor and $c(d)\equiv d \<  \tr \rho_0^2 \>-1$ for the super-fidelity. This finally gives the following 
lower bounds on the decay times:
\begin{eqnarray}
&&\tau_{dec}\equiv \left[8g^2N_{uno}d_S\left(d-\<\tr\rho_0^2\>\right)\right]^{-1/2}\sim \left[8g^2N_{uno}d_Sd\right]^{-1/2},\\
&&\tau_{fid}\equiv \left[8g^2N_{mac}d_S\left(d\<\tr\rho_0^2\>-1\right)\right]^{-1/2}\sim\left[8g^2N_{mac}d_S\right]^{-1/2},
\end{eqnarray}
where the simplified estimates are in the limit of a large local dimension of the apparatus $d$. The above times are clearly within the short-time approximation range \eqref{stlim}.
It is interesting to note that the above
separation of physical time-scales of decoherence and information accumulation is, on the mathematical level, a consequence of 
a simple symmetry difference in the initial formulas \eqref{eq:gamma_appendix} and  \eqref{Sfid}. The latter has an additional SWAP operator $\mathbb V$ under
the trace. 

As mentioned, the above timescales are only lower bounds for the actual times of the noise floor approach. Perhaps they can be tightened using different analytical tools,
but for the purpose of this work we will use the above $\tau_{dec},\tau_{fid}$. 
%----------FIGURE -------
%-------------------------
%-------------------------
%\newpage
\begin{figure}
\centering
\begin{minipage}{.5\columnwidth}
  \centering
\includegraphics[width=0.9\textwidth]{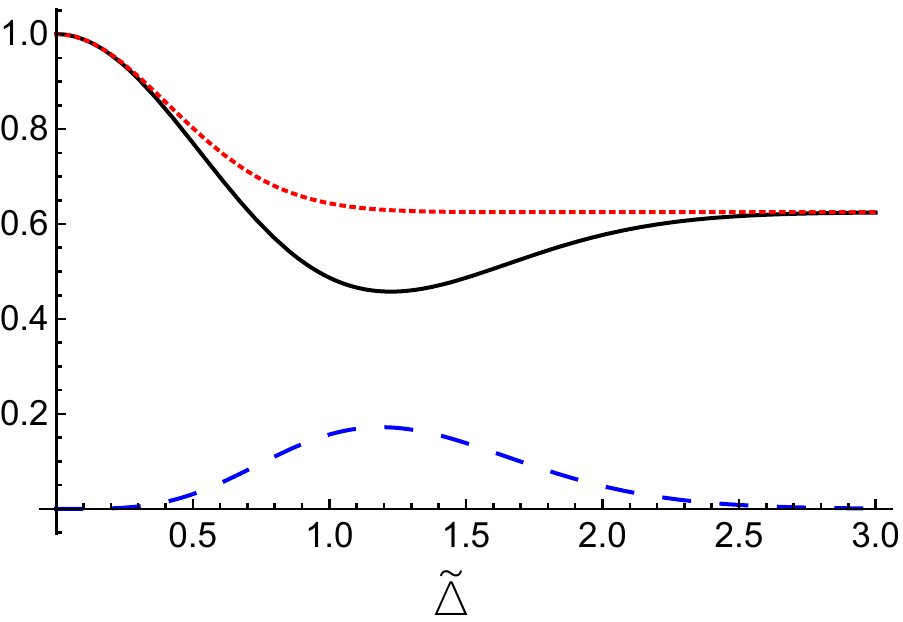}
  \label{fig:deco-ansatz-d2}
    \centerline{(a) $d=2$ }
\end{minipage}%
\begin{minipage}{.5\columnwidth}
  \centering
\includegraphics[width=0.9\textwidth]{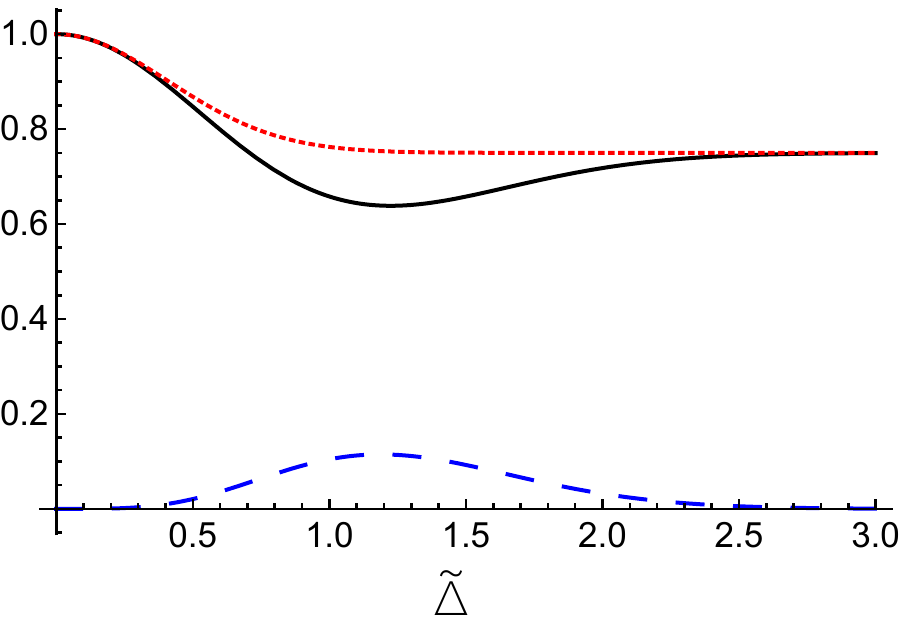}
  \label{fig:fid-ansatz-d2}
    \centerline{(b) $d=2$}
\end{minipage}\\
\vspace{.5cm}
\begin{minipage}{.5\columnwidth}
  \centering
\includegraphics[width=0.9\textwidth]{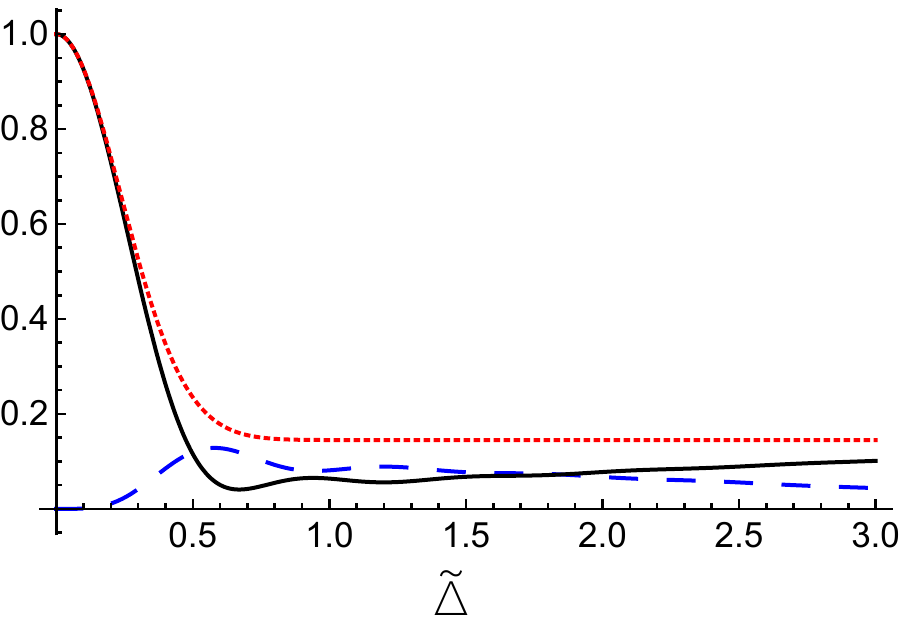}
  \label{fig:deco-ansatz-d8}
    \centerline{(c) $d=8$}
\end{minipage}%
\begin{minipage}{.5\columnwidth}
  \centering
\includegraphics[width=.9\textwidth]{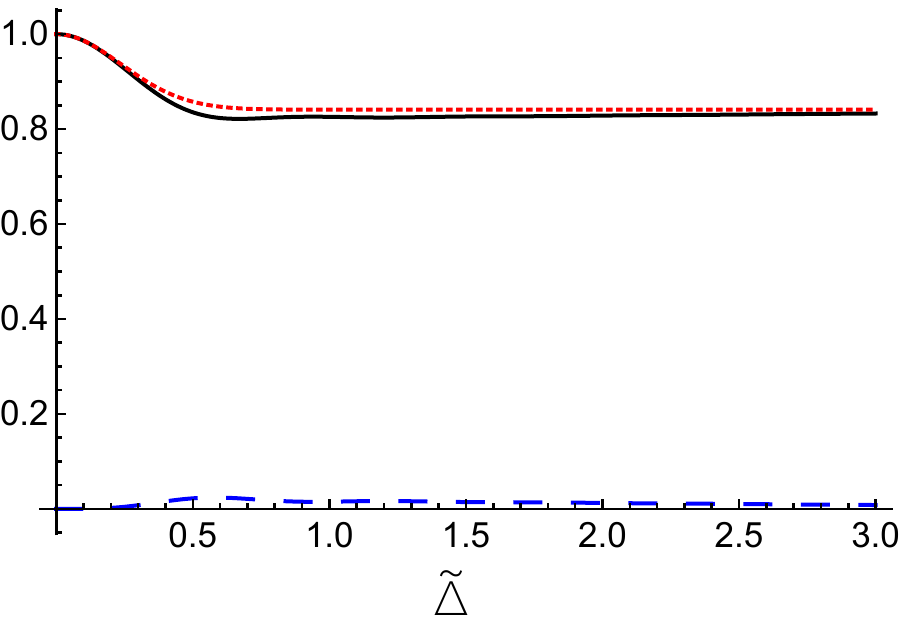}
  \label{fig:fid-ansatz-d8}
    \centerline{(d) $d=8$ }
\end{minipage}
\caption{(Color online). In figures (a) and (c) the dotted (red) curve represents the Ansatz
$\tilde{\Gamma}(t)$, the solid (black) line represents the exact decoherence factor $\<\Gamma(t)\>$, and the dashed (blue) curve the difference $\tilde{\Gamma}(t)-\<\Gamma(t)\>$.
Similarly, in (b) and (d) the dotted (red) curve represents the Ansatz $\tilde{G}(t)$, the solid (black) line represents the exact averaged super-fidelity $\<G(t)\>$, and the dashed (blue)
curve the difference $\tilde{G}(t)-\<G(t)\>$.
The timescale is given in units of $\tilde{\Delta}=(a-a')/\sqrt{\eta_E}$, the Bures average purity is used throughout, and the exact integration is done assuming $d_S=2$.}
\label{fig:ansatzfig}
\end{figure}
%----------END FIGURE ----
%-------------------------

\renewcommand{\theequation}{C\arabic{equation}}
\renewcommand{\thefigure}{C\arabic{figure}}

\section{Upperbounding the exact expression for $\<f_t(\pmb a,\pmb\lambda)\>$}\label{ansatz}
This Section is complimentary to the main line of reasoning and is not necessary for understanding it. 
The exact expression \eqref{av}
is algebraically  complicated  and a simplified expression,
reproducing short- and long-time behavior and preferably upper bounding it would be desirable.
Here we attempt a construction of  such an upper bound.
For short times, we may neglect in the average $\< f_t(\pmb a,\pmb{\lambda})\> = p(d,\tilde{\Delta}_t) \E^{-\tilde{\Delta}_t^2}$ 
all terms except the first two giving $\< f_t(\pmb a,\pmb{\lambda})\>\lesssim [d(d-1)/2] \E^{-(d+1)\tilde\Delta_t^2}$ (see Eq. \eqref{ft}).
On the other hand, for large times the highest order term
$\tilde\Delta_t^{2(2d-3)}$ dominates and it always comes with a negative coefficient as can be seen from Eq. \eqref{av}.
These observations suggest  the following  Ansatz:
\begin{eqnarray}
&&\tilde{\Gamma}_{aa'}(t) \equiv \frac{1+\<\tr \rho^2_0\>}{d+1} + \frac{d-\<\tr \rho^2_0\> }{d+1}\E^{-(d+1)\tilde{\Delta}_t^2}, \label{eq:ansatzG}\\
&&\tilde{G}_{aa'}(t) \equiv \<S_{lin}(\rho_0)\>+\frac{1+\<\tr \rho^2_0\>}{d+1} + \frac{d\<\tr \rho^2_0\> -1}{d+1}\E^{-(d+1)\tilde{\Delta}_t^2},\label{eq:ansatzF}
%\end{split}
\end{eqnarray}
for the decoherence and the superfidelity factors respectively.
By construction, the expressions (\ref{eq:ansatzG},\ref{eq:ansatzF}) reproduce the correct short time behavior
as $\< \Gamma_{aa'}(t) \>= \tilde{\Gamma}_{aa'}(t) + O(\tilde\Delta_t^4)$, and similarly $\< G_{aa'}(t) \>$,
for $t\ll\tau_{aa'}$, where $\tau_{aa'}$ is the characteristic time of the Gaussian decay in (\ref{eq:ansatzG},\ref{eq:ansatzF}):
\be\label{tdec}
\tau_{aa'}\equiv \frac{\sqrt{\eta_E}}{\sqrt{d+1}|a-a'|}.
\ee

%\end{appendices}

They are also upper bounds since $1-x\leq\E^{-x}$ (cf. (\ref{pshort})). Similarly, for long times $t\gg\tau_{aa'}$,
(\ref{eq:ansatzG},\ref{eq:ansatzF}) reproduce
the correct white noise factors and also upper bound the exact averages, since $p(d,\tilde{\Delta}_t)$ is then negative.
These two facts together suggest that  (\ref{eq:ansatzG},\ref{eq:ansatzF}) may be upper bounds
for all times $t$ and dimensions $d$. Unfortunately we were unable to prove it analytically, but numerical evidence for $d\leq 20$ suggests that it is indeed so.
In Fig.~\ref{fig:ansatzfig} we present some sample plots of both the exact expressions
and (\ref{eq:ansatzG},\ref{eq:ansatzF}) together with the errors.
The price to pay for working with the simplified expressions (\ref{eq:ansatzG},\ref{eq:ansatzF}) is that one loses
the important physical information on non-Markovianity, reflected by the non-monotonic behaviour of the exact averages for low rations $N/d$ .
Also the exact functions approach their asymptotic limits from below, signalizing recovery of coherences/loss of information in the environment,
while (\ref{eq:ansatzG},\ref{eq:ansatzF}) approach them from above. However, if
one is solely interested in an SBS formation, an upper bound decaying to the correct noise level is just enough.


\begin{thebibliography}{44}
\bibitem[{not()}]{nota}
{\bibinfo{title}{For pure $\rho_0$ both fidelity and decoherence factor
  become equal so that vanishing of the latter already implies an {SBS}
  formation.}}
\end{thebibliography}
\end{document}